\documentclass[preprint,11pt]{elsarticle}
\usepackage{amsmath}
\usepackage{amsfonts}   
\usepackage{amssymb}    
\usepackage{amsthm}
\usepackage{subfigure}

\makeatletter
\def\ps@headings{%
\def\@oddhead{\mbox{}\scriptsize\rightmark \hfil \thepage}%
\def\@evenhead{\scriptsize\thepage \hfil \leftmark\mbox{}}%
\def\@oddfoot{}%
\def\@evenfoot{}}
\makeatother
\journal{arXiv}

\begin{document}

\begin{frontmatter}

\newtheorem{thm}{Theorem}[section]
\newtheorem{cor}[thm]{Corollary}
\newtheorem{lem}[thm]{Lemma}
\newtheorem{prop}[thm]{Proposition}

\newtheorem{rem}[thm]{Remark}

\newtheorem{ex}[thm]{Example}

\newtheorem{defi}[thm]{Definition}

%
\title{A Note on the Computational Complexity of Unsmoothened Vertex Attack Tolerance}


\author{Gunes Ercal}

\address{Southern Illinois University Edwardsville}

\begin{abstract}
We have previously introduced vertex attack tolerance (VAT) and unsmoothened VAT (UVAT), denoted respectively as $\tau(G) =  \min_{S \subset V} \frac{|S|}{|V-S-C_{max}(V-S)|+1}$ and $\hat{\tau}(G) =  \min_{S \subset V} \frac{|S|}{|V-S-C_{max}(V-S)|}$, where  $C_{max}(V-S)$ is the largest connected component in $V-S$, as appropriate mathematical measures of resilience in the face of targeted node attacks for arbitrary degree networks.  Here we prove the hardness of approximating $\hat{\tau}$ under various plausible computational complexity hypotheses.
\end{abstract}

\end{frontmatter}

\section{Definitions and Preliminaries}
Given a connected, undirected graph $G = (V,E)$, the Vertex Attack Tolerance of $G$ is denoted by $\tau(G)$ defined as follows:\cite{CASResSF,MattaBE14,spectralVAT}
\begin{equation*}
\tau(G) = \min_{S \subset V, S \neq \emptyset} \{\frac{|S|}{|V-S-C_{max}(V-S)|+1 } \}
\end{equation*}
where $C_{max}(V-S)$ is the largest connected component in $V-S$.  As in \cite{spectralVAT}, we refer to connected, undirected graphs $G = (V,E)$ with more than one node ($|V| \ge 2$) as \emph{non-trivial}.
\begin{rem}\label{vatrange}\cite{spectralVAT}
For nontrivial $G = (V,E)$, $0 < \tau(G) \le 1$.
\end{rem}

VAT was originally introduced as $\hat{\tau}$ (UVAT for ``unsmoothened VAT''), of which $\tau$ is a smoothened variation, defined as follows\cite{CASResSF,MattaBE14}:
\begin{equation*}
\hat{\tau} = \min_{S \subset V, S \neq \emptyset} \{\frac{|S|}{|V-S-C_{max}(V-S)|}\}
\end{equation*}
where $C_{max}(V-S)$ is the largest connected component in $V-S$.  Note that for any graph $G = (V,E)$ such that $G$ is not a clique, the pair of nodes $u,v$ which are not adjacent may be disconnected by attacking all of the other $n-2$ nodes.  However, for cliques $K_n$, no such pair exists.  Therefore:
\begin{rem}
$\hat{\tau}$ is undefined for cliques $K_n$ and defined for all other graphs.  Moreover, when $G = (V,E) \neq K_n, \hat{\tau}(G) \le n-2$.  Furthermore, when $G = (V,E)$ is connected, $\hat{\tau}(G) > 0$.  Therefore, when $G = (V,E)$ is connected and not complete, $S(\hat{\tau})$ is a vertex separator.
\end{rem}
For notational convenience: For any graph $G = (V,E)$, and any real function $f$ defined on subsets of $V$, if $h = \min_{S \subset V} f$, we define $h_S(G) = f(S)$ and $S(h(G)) = argmin_{S \subset V} f(S)$.  In particular, when $h$ is a resilience measure on a graph, then $S(h)$ denotes the critical attack set.

We refer to the optimization problem corresponding to computing $\tau(G)$ and $\hat{\tau}(G)$ as simply VAT and UVAT, respectively.  It is assumed that any approximation algorithm for UVAT returns a candidate critical attack set that is a valid vertex separator when the input is not a clique (as finding some vertex separator is easy).

The reduction in this work extends the techniques for the NP-Hardness proof for the vertex integrity of co-bipartite graphs presented in \cite{integrityHard}.  Similarly, our computational hardness results for VAT and other measures involve reductions with the Balanced Complete Bipartite Subgraph problem (BCBS).  The BCBS problem is defined as:
\begin{defi}
\emph{Instance:} A balanced bipartite graph $G = (V_1, V_2, E)$ with $n = |V_1| = |V_2|$ and an integer $0 < k \le n$.
\emph{Question:} Does there exist $A \subset V_1$ and $B \subset V_2$ such that $|A| = |B| = k$ and $(A,B)$ form a $k \times k$ complete bipartite graph? 
\end{defi}
The maximization version of the problem can be referred to as MAX-BCBS.  The following three theorems regard the hardness of approximating MAX-BCBS under various plausible complexity theoretic assumptions:

\begin{thm}\cite{BCBSinapprox}\label{thm:BCBSHard1}
It is NP-hard to approximate the MAX-BCBS problem within a constant factor if it is NP-hard to approximate the maximum clique problem within a factor of $n/2^{c\sqrt{\log{n}}}$ for some small enough $c > 0$.
\end{thm}

\begin{thm}\cite{noPTASKhot}\label{thm:BCBSHard2}
Let $\epsilon > 0$ be an arbitrarily small constant.  Assume that SAT does not have a probabilistic algorithm that runs in time $2^{n^\epsilon}$ on an instance of size $n$.  Then there is no polynomial time (possibly randomized) algorithm for MAX-BCBS that achieves an approximation ratio of $N^{\epsilon'}$ on graphs of size $N$ where $\epsilon' = \frac{1}{2^{O(1/\epsilon \log{(1/\epsilon)})}}$.
\end{thm}

%
\begin{thm}\cite{BCBSinapprox4}\label{thm:BCBSHard4}
MAX-BCBS is R4SAT-Hard to approximate within a factor of $n^\delta$ where $n$ is the number of vertices in the input graph, and $0 < \delta < 1$ is some constant.  More specifically, under the random 4-SAT hardness hypothesis: There exists two constants $\epsilon_1 > \epsilon_2 > 0$ such that no efficient algorithm is able to distinguish between bipartite graphs $G = (V_1,V_2,E)$ with $|V_1|=|V_2|=n$ which have a clique of size $\ge (n/16)^2(1+\epsilon_1)$ and those in which all bipartite cliques are of size $\le (n/16)^2(1+\epsilon_2)$.
\end{thm}

\section{Results}
Our main theorem is as follows:
\begin{thm}\label{thm:uvathard}
All of the following statements hold even when UVAT is restricted to co-bipartite graphs.
\begin{itemize}
\item[(I)] It is NP-Hard to approximate UVAT within a constant factor if it is NP-hard to approximate the maximum clique problem within a factor of $n/2^{c\sqrt{\log{n}}}$ for some small enough $c > 0$.
\item[(II)] Let $\epsilon, \epsilon'$ be as in Theorem \ref{thm:BCBSHard2}.  If SAT has no probabilistic algorithm that runs in time $2^{n^\epsilon}$ on instances of size $n$, then there is no polynomial time (possibly randomized) algorithm for UVAT that achieves an approximation ratio of $N^{\epsilon'}$ on graphs of size $N$
\item[(III)] UVAT is R4SAT-Hard to approximate within a factor of $n^\delta$ where $n$ is the number of vertices in the input graph, and $0 < \delta < 1$ is some constant.
\end{itemize}
\end{thm}

The theorem follows directly from part (III) of the following Lemma and Theorems \ref{thm:BCBSHard1}, \ref{thm:BCBSHard2}, \ref{thm:BCBSHard4}.
\begin{lem}\label{lem:vatcobp}
Let $G = (V_1, V_2, E)$ with $|V_1|=|V_2|=n$ be a bipartite graph with $E \neq \emptyset$, and let $\overline{G} = (V_1, V_2, \overline{E})$ be the co-bipartite complement of $G$.  Let $BK(G) = \{ (A,B) | A \times B \text{ is a bipartite clique in } G \text{ with } |A| \le |B| \}$.  Moreover, let $BBK(G) = \{ (A,B) \in V_1 \times V_2 | A \times B \text{ is a bipartite clique of } G \text{ with } |A| = |B| \}$, and let $(\hat{A},\hat{B}) = argmax_{(A,B) \in BBK(G)} |A|$ be the maximum balanced bipartite clique of $G$ with corresponding size $k = |\hat{A}|$.  Then, the following hold:
\begin{itemize}
\item[(I)] $\hat{\tau}(\overline{G}) = \min_{(A,B) \in BK(G)} \frac{2n-|A|-|B|}{|A|} =  \min_{(A,B) \in BK(G)} \frac{2n-|B|}{|A|} - 1$
\item[(II)] $\frac{n}{k} - 1 \le \hat{\tau}(\overline{G}) \le 2(\frac{n}{k} - 1)$
\item[(III)] If UVAT can be approximated to factor $\alpha$ in polynomial time, then MAX-BCBS can be approximated to factor $2\alpha$ in polynomial time, even when restricted to co-bipartite graphs.
\end{itemize}
\end{lem}

\begin{proof}[Proof of Lemma \ref{lem:vatcobp}]
Let $S = S(\tau(\overline{G}))$, $U = S(\hat{\tau}(\overline{G}))$,  $R = S(I(\overline{G}))$, and  $C = S(T(\overline{G}))$ be the critical attack sets corresponding to $\tau$, $\hat{\tau}$, $I$, and $T$ for $\overline{G}$, respectively.  Furthermore, let $S_i = V_i \cap S$, $U_i = V_i \cap U$, $R = V_i \cap R$, and $C_i = V_i \cap C$.  For $X \in \{ S, U, R, C \}$, let $A_X = \min \{V_1 - X_1, V_2 - X_2\}$ and  $B_X = \max \{V_1 - X_1, V_2 - X_2\}$.

Note that $\overline{G}$ is not a clique as $E \neq \emptyset$.  Moreover, because $V_1$ and $V_2$ must both be cliques in $\overline{G}$, $A_X$ and $B_X$ must each be cliques in $\overline{G}$, for any $X \in \{ S, U, R, C \}$.  Namely, the removal of $X$ results in exactly two cliques $A_X$ and $B_X$ in $\overline{G}$.  Clearly, there can be no edge between $A_X$ and $B_X$ in $\overline{G}$ as such an edge would have remained upon the removal of $X$.  Therefore, $(A_X,B_X)$ forms a bipartite clique in $G$.  Part (I) of the lemma now follows from the definitions of $\hat{\tau}$ and the fact that $|A_X| \le |B_X|$.

Now note that for any $(A,B) \in BK(G)$, any subset $B_A \subset B$ such that $|B_A| = |A|$ forms a balanced bipartite clique with $A$.  Also clearly, $BBK(G) \subset BK(G)$.  Therefore, by (I) and fact that $|A_X| \le |B_X| \le n$, (II) follows as well.

For part (III): Let $M$ be an algorithm that gives a constant factor approximation for UVAT with approximation factor $\alpha > 1$.  Let $q$ such that $\hat{\tau}(\overline{G}) \le q \le \alpha\hat{\tau}(\overline{G})$ be the approximation to $\hat{\tau}$ computed via $M$ on the input.  Simplifying and rearranging Lemma \ref{lem:vatcobp} part (II.b):
\begin{equation}
\frac{n}{q+1} \le k \le \frac{n}{1+q/(2\alpha)}
\end{equation}
Similarly, let $r = (\frac{n}{q+1})/(\frac{n}{1+q/(2\alpha)})$ denote the ratio between the right hand side and left hand side of the inequality, so:
\begin{equation}
r = \frac{q+1}{1+q/(2\alpha)}
\end{equation}
If $r > 2\alpha$ then $1 > 2\alpha$ resulting in a contradiction.  Therefore, 
\begin{equation}
\frac{n}{q+1} \le k \le 2\alpha\frac{n}{q+1}
\end{equation}
And, $\frac{n}{q+1}$ is thus a $\frac{1}{2\alpha}$ approximation for the MAX-BCBS problem with corresponding approximation ratio $2\alpha$.
\end{proof}

\bibliographystyle{plain}
\bibliography{VAT}

\end{document}